\documentclass[envcountsame,envcountsect,runningheads]{llncs} %

\usepackage{graphicx}
\usepackage{latexsym}
\usepackage{amsmath,amssymb,amstext}
\usepackage{comment}
\usepackage{pifont}

\usepackage{enumerate}
\usepackage{color}
\usepackage{wrapfig}
\usepackage{mathtools}

\makeatletter
\def\moverlay{\mathpalette\mov@rlay}
\def\mov@rlay#1#2{\leavevmode\vtop{%
   \baselineskip\z@skip \lineskiplimit-\maxdimen
   \ialign{\hfil$#1##$\hfil\cr#2\crcr}}}
\makeatother

\newcommand{\stdarraystretch}{
  \renewcommand{\arraystretch}{1.3}
  \setlength{\tabcolsep}{1ex}
}

\newcommand{\appsection}[1]{
  \let\oldthesection\thesection
  \renewcommand{\thesection}{Appendix \oldthesection}
  \section{#1}
  \let\thesection\oldthesection
}

\newcommand{\ignore}[1]{}

\renewcommand{\mit}{\mathit}

\newcommand{\msf}{\mathsf}

\newcommand{\SN}{\msf{SN}}
\newcommand{\SNr}[1]{\SN_{#1}}

\newcommand{\SNrs}[1]{\funap{\SN_{#1}}}
\newcommand{\SNi}{\SN^\infty}

\newcommand{\WN}{\msf{WN}}
\newcommand{\WNr}[1]{\WN_{#1}}

\newcommand{\WNrs}[1]{\funap{\WN_{#1}}}
\newcommand{\WNi}{\WN^\infty}

\newcommand{\SNw}{\SN^\omega}

\newcommand{\CR}{\msf{CR}}
\newcommand{\CRr}[1]{\CR_{#1}}
\newcommand{\CRrs}[1]{\funap{\CRr{#1}}}

\newcommand{\gCR}{\msf{grCR}}

\newcommand{\WCR}{\msf{WCR}}
\newcommand{\WCRr}[1]{\WCR_{#1}}
\newcommand{\WCRrs}[1]{\funap{\WCRr{#1}}}

\newcommand{\gWCR}{\msf{grWCR}}

\newcommand{\DP}{\msf{DP}}
\newcommand{\DPmin}{\msf{DP}^{\msf{min}}}

\newcommand{\tp}{\msf{top}}
\newcommand{\SNdp}[2]{\SN(#1_{\tp}/#2)}
\newcommand{\SNdps}[3]{\SN(#3,#1_{\tp}/#2)}
\newcommand{\SNdpm}[2]{\SN(#1_{\tp}/_{\msf{min}}\,#2)}

\newcommand{\sred}{{\rightarrow}}
\newcommand{\red}{\mathrel{\sred}}
\newcommand{\redrat}[2]{\mathrel{\sred_{#1,#2}}}
\newcommand{\redr}[1]{\mathrel{\sred_{#1}}}
\newcommand{\redrroot}[1]{\redrat{#1}{\posemp}}

\newcommand{\sredi}{{\leftarrow}}
\newcommand{\redi}{\mathrel{\sredi}}
\newcommand{\smred}{{\twoheadrightarrow}}
\newcommand{\mred}{\mathrel{\smred}}
\newcommand{\smredi}{{\twoheadleftarrow}}
\newcommand{\mredi}{\mathrel{\smredi}}
\newcommand{\pairlft}{{\langle}}                %
\newcommand{\pairrgt}{{\rangle}}                %
\newcommand{\pairsep}{{,\,}}                    %
\newcommand{\pairstr}[1]{\pairlft#1\pairrgt}    %
\newcommand{\pair}[2]{\pairstr{#1\pairsep#2}}   %
\newcommand{\triple}[2]{\pair{#1\pairsep#2}}    %
\newcommand{\quadruple}[2]{\triple{#1\pairsep#2}} 

\newcommand{\sjoin}{{\cup}}
\newcommand{\join}{\mathbin{\sjoin}}

\newcommand{\sfunin}{{:}}
\newcommand{\funin}{\mathrel{\sfunin}}
\newcommand{\powerset}[1]{{\boldsymbol 2}^{#1}}

\newcommand{\setemp}{{\varnothing}}%

\newcommand{\mybind}[3]{#1#2.\:#3}
\newcommand{\myex}{\mybind{\exists}}
\newcommand{\myall}{\mybind{\forall}}

\newcommand{\nat}{\mathbb N}
\newcommand{\zz}{\mathbb Z}

\catcode`\@=11

\newcommand{\mywash}[2]{\setbox0=\hbox{$\m@th#1{#2}$}\wd0=0pt\box0}
\catcode`\@=12

\newcommand{\asig}{\Sigma}

\newcommand{\avars}{\mathcal{X}}
\newcommand{\ster}{\mit{Ter}}
\newcommand{\ter}{\funap{\ster}}

\newcommand{\svar}{\mit{Var}}

\newcommand{\vars}{\funap{\svar}}

\newcommand{\atrs}{R}
\newcommand{\btrs}{S}

\newcommand{\posemp}{\epsilon}
\newcommand{\apos}{p}

\newcommand{\pos}{\funap{\mathcal{P}\!os}}

\newcommand{\binap}[3]{#2\mathbin{#1}#3}
\newcommand{\funap}[2]{#1(#2)}

\newcommand{\bfunap}[3]{\funap{#1}{#2,#3}}
\newcommand{\tfunap}[4]{\funap{#1}{#2,#3,#4}}

\newcommand{\where}{\mathrel{|}}

\newcommand{\sdefdby}{{:}{=}}
\newcommand{\defdby}{\mathrel{\sdefdby}}

\renewcommand{\implies}{\Rightarrow}

\newcommand{\sarity}{\sharp}%
\newcommand{\arity}{\funap{\sarity}}

\newcommand{\autstates}{Q}

\newcommand{\astate}{q}

\newcommand{\asubst}{\sigma}
\newcommand{\subst}[2]{#2#1}

\newcommand{\strff}{\msf} %

\newcommand{\lstemp}{\varepsilon} %
\newcommand{\lstconcat}[2]{#1 #2}

\newcommand{\lstcat}[2]{#1#2} %

\newcommand{\lstlength}[1]{|#1|}

\newcommand{\sstrcns}{\strff{:}}

\newcommand{\strcnsd}[1]{\binap{\sstrcns}}

\newcommand{\punc}[1]{\:\text{#1}}

\newcommand{\acontext}{C}
\newcommand{\contexthole}{\Box}
\newcommand{\contextfill}[2]{#1[#2]}

\newcommand{\cpi}[2]{\mathrm{\Pi}^{#1}_{#2}}
\newcommand{\csig}[2]{\mathrm{\Sigma}^{#1}_{#2}}
\newcommand{\tm}{\msf{M}}
\newcommand{\tmstates}{Q}
\newcommand{\tmsig}{\Gamma}
\newcommand{\stmtrans}{\delta}
\newcommand{\tmtrans}{\bfunap{\stmtrans}}
\newcommand{\tmiblank}{\triangleright}

\newcommand{\stmblank}{\Box}
\newcommand{\tmblank}{\funap{\stmblank}}
\newcommand{\tmL}{L}
\newcommand{\tmR}{R}
\newcommand{\stmmap}{\phi}
\newcommand{\tmmap}{\funap{\stmmap}}
\newcommand{\tmterms}[1]{\ster_{#1}}
\newcommand{\stmint}{\Phi}
\newcommand{\tmint}{\funap{\stmint}}
\newcommand{\tmstart}{\astate_0}
\newcommand{\stmtape}{\mit{tape}}
\newcommand{\tmtape}{\funap{\stmtape}}
\newcommand{\tmconf}[1]{\mathcal{C}\!\mit{onf}_{#1}}

\newcommand{\tmfinal}[2]{\funap{\mit{final}_{#2}}{{#1}}}

\newcommand{\stmzer}{\msf{0}}
\newcommand{\tmzer}{\funap{\stmzer}}
\newcommand{\stmsucc}{\msf{S}}
\newcommand{\tmsucc}{\funap{\stmsucc}}

\newcommand{\tmT}{\msf{T}}

\newcommand{\stmpeb}{\bullet}
\newcommand{\tmpeb}{\funap{\stmpeb}}

\newcommand{\tmtrs}[1]{\atrs_{#1}}
\newcommand{\tmtrspeb}[1]{\atrs_{#1}^{\stmpeb}}
\newcommand{\tmstep}{\to_\tm}

\newcommand{\stmfun}[1]{f_{#1}}
\newcommand{\tmfun}[1]{\funap{\stmfun{#1}}}
\newcommand{\tmrel}[1]{\mathrel{>_{#1}}}

\newcommand{\sfs}{\stmsucc}
\newcommand{\fs}{\funap{\sfs}}
\newcommand{\sfc}{\msf{c}}
\newcommand{\fc}{\funap{\sfc}}

\newcommand{\pickn}{\msf{pickn}}
\newcommand{\picknok}{\funap{\msf{ok}}}
\newcommand{\trspickn}{\atrs_{\pickn}}

\newcommand{\pto}{\rightharpoonup}

\newcommand{\WF}{\msf{WF}}

\newcommand{\first}{\funap{\mit{first}}}
\newcommand{\last}{\funap{\mit{last}}}

\usepackage{diagrams}
\newcommand{\REC}{\mbox{REC}}
\newcommand{\lth}{{\rm lth}}
\newcommand{\Seq}{\mbox{Seq}}
\newcommand{\forallt}{\forall^1}
\newcommand{\existst}{\exists^1}
\newcommand{\lem}{\leq_m}
\newcommand{\weg}[1]{}
\def\phi{\varphi}

\begin{document}

\title{Degrees of Undecidability in Rewriting}

\author{
  J\"{o}rg Endrullis\inst{1}
  \and Herman Geuvers\inst{2,3}
  \and Hans Zantema\inst{3,2}
}
\institute{
    Vrije Universiteit Amsterdam, The Netherlands\\
    \email{joerg@few.vu.nl}
  \and
    Radboud Universiteit Nijmegen, The Netherlands\\
    \email{herman@cs.ru.nl}
  \and
    Technische Universiteit Eindhoven, The Netherlands\\
    \email{h.zantema@tue.nl}
}
\maketitle

\begin{abstract}
Undecidability of various properties of first order term rewriting
systems is well-known.  An undecidable property can be classified by
the complexity of the formula defining it. This gives rise to a
hierarchy of distinct levels of undecidability, starting from the
arithmetical hierarchy classifying properties using first order
arithmetical formulas and continuing into the analytic hierarchy,
where also quantification over function variables is allowed.

In this paper we consider properties of first order term rewriting systems and 
classify them in this hierarchy. 
Weak and strong normalization for
single terms turn out to be $\csig{0}{1}$-complete, while their uniform versions 
as well as dependency pair problems with minimality flag are $\cpi{0}{2}$-complete.
We find that confluence is $\cpi{0}{2}$-complete both for single terms and uniform.
Unexpectedly weak confluence for ground terms turns out to be harder
than weak confluence for open terms.
The former property is $\cpi{0}{2}$-complete while the latter is $\csig{0}{1}$-complete (and thereby recursively enumerable).

The most surprising result is on dependency pair problems without
minimality flag: we prove this to be $\cpi{1}{1}$-complete, which
means that this property exceeds the arithmetical hierarchy and is
essentially analytic.
\end{abstract}

\section{Introduction}

In classical computability theory a property
$P \subseteq \nat$ is called \emph{decidable}
iff there exists a Turing machine which
for every input $x \in \nat$ outputs $0$ if $x\in P$ and $1$ if $x\notin P$.
The complexity of decidable properties is usually
defined in terms of the time (or space) consumption
of a Turing machine that decides the property;
the respective hierarchies (linear, polynomial, exponential,\ldots) are well-known.
Likewise, but less known, the undecidable properties can be classified into a
hierarchy of growing complexity.  
The arithmetical and the analytical hierarchy establish such a
classification of undecidable properties by the complexity of
predicate logic formulas that define them, which in turn is defined as
the number of quantifier alternations of its prenex normal form.
The \emph{arithmetical hierarchy} is based on first order formulas,
that is, quantification is restricted to number quantifiers, function
or set quantification is not allowed; its classes are denoted
$\cpi{0}{n}$ and $\csig{0}{n}$ for $n \in \nat$.  The lowest level of
the hierarchy, the classes $\cpi{0}{0}$ and $\csig{0}{0}$, consists of
the decidable relations (for which there is a total computable
function that decides it).  Then the classes $\cpi{0}{n}$ and
$\csig{0}{n}$ for $n \ge 1$ are inductively defined by allowing
additional universal and existential quantifiers to define the
properties. For example, if $P(x,y,z)$ is a decidable property, then
$\exists x\,P(x,y,z)$ is in $\csig{0}{1}$ and $\forall y\,\exists x\,
P(x,y,z)$ is in $\cpi{0}{2}$. 
In other words, a relation belongs to the class $\cpi{0}{n}$ for
$n \in \nat$ of the arithmetical hierarchy if it can be defined by a
first order formula (in prenex normal form), which has $n$
quantifiers, starting with a universal quantifier.  Likewise a
relation is in $\csig{0}{n}$ if the formula starts with an existential
quantifier.  The class $\csig{0}{1}$ is the class of {\em recursively
enumerable\/} (or semi-decidable) relations; the
special halting problem is in this class.  The
general halting problem is in the class $\cpi{0}{2}$.

The \emph{analytical hierarchy} continues the classification of
relations by second order formulas,
allowing for function quantifiers.
Its classes are denoted $\cpi{1}{n}$ and $\csig{1}{n}$ for $n \in \nat$.
The lowest level of the analytical hierarchy, that is, 
the classes $\cpi{1}{0}$ and $\csig{1}{0}$ consist of all arithmetical relations.
The classes $\cpi{1}{n}$ and $\csig{1}{n}$ for $n \ge 1$
are defined inductively, each time adding an universal ($\myall{\alpha \funin \nat \to \nat}{\phi}$)
or existential function quantifier ($\myex{\alpha \funin \nat \to \nat}{\phi}$), respecively.
For example the class $\cpi{1}{1}$ consists of relations
which can be defined by $\myall{\alpha \funin \nat \to \nat}{\phi}$ where $\phi$ is an arithmetical relation.

\paragraph{Our Contribution}

We investigate the arithymetic complexity,
of various properties of first order TRSs:
\begin{itemize}
 \item termination or strong normalization ($\SN$),
 \item weak normalization ($\WN$),
 \item confluence ($\CR$) and ground confluence ($\gCR$),
 \item weak confluence ($\WCR$) and weak ground confluence ($\gWCR$),
 \item finiteness of dependency pair problems ($\DP$), and
 \item finiteness of dependency pair problems with minimality flag ($\DPmin$).
\end{itemize}

\begin{figure}[h!]
\vspace{-2ex}
  \begin{center}
    \stdarraystretch
    \begin{tabular}{|c|c|c|c|c|c|c|}
    \hline
    & $\SN$ & $\WN$ & $\CR$ & $\gCR$ & $\WCR$ & $\gWCR$\\
    \hline
    uniform & $\cpi{0}{2}$ & $\cpi{0}{2}$ & $\cpi{0}{2}$ & $\cpi{0}{2}$ & $\csig{0}{1}$ & $\cpi{0}{2}$\\
    \hline
    single term & $\csig{0}{1}$ & $\csig{0}{1}$ & $\cpi{0}{2}$ & $\cpi{0}{2}$ & $\csig{0}{1}$ & $\csig{0}{1}$\\
    \hline
    \end{tabular}
    \quad 
    \begin{tabular}{|c|c|c|}
    \hline
    & $\DP$ & $\DPmin$\\
    \hline
    uniform & $\cpi{1}{1}$ & $\cpi{0}{2}$\\
    \hline
    single term & $\cpi{1}{1}$ & $\csig{0}{1}$\\
    \hline
    \end{tabular}
  \end{center}
\vspace{-2ex}
\caption{Degrees of undecidability}
\label{fig:complexities}
\vspace{-2ex}
\end{figure}

While undecidability of these concepts is
folklore \cite{GMOZ02a}
their degree of undecidability, their
precise hardness, has hardly been studied, with the exception
of \cite{HuetLankford} who study the Turing degree of
termination. Turing degrees give a classification of undecidable
properties in terms of their computational `hardness' which is
independent of the syntactic form of a predicate that describes
it. Their is a connection between the Turing degree of a property and
its place in teh arithmetic hierrachy, so most of the proofs
of \cite{HuetLankford} can be carry over to our setting. As we use a
different translation from Turing machines to TRSs, we do not use the
results or proofs of \cite{HuetLankford}.

In this paper we pinpoint the precise complexities of these properties in terms of the arithmetic (and analytic) hierachy, see Figure~\ref{fig:complexities}; we study these properties 
\emph{uniformly} for all terms (as a system property) as well as
for \emph{single terms}.

We find that the standard TRS properties $\SN$, $\WN$, $\CR$, $\WCR$
reside within the classes $\cpi{0}{2}$ and $\csig{0}{1}$
of the arithmetical hierarchy, for the uniform and single term versions, respectively.
That is, they are of a low degree of undecidability,
being at most as hard as the general halting problem.

Unexpectedly we find 
that weak ground confluence is a harder decision problem than weak confluence.
While weak confluence is $\csig{0}{1}$-complete and therefore recursively enumerable
it turns out that weak ground confluence $\cpi{0}{2}$-complete.

Surprisingly, it turns out that dependency pair problems are of a much
higher degree of undecidability: they exceed the whole arithmetical
hierarchy and thereby first order predicate logic.  In particular
we show that dependency pair problems are $\cpi{1}{1}$-complete, a
class within the analytical hierarchy with one universal function
quantifier. So although dependency pair problems are invented for
proving termination, the complexity of general dependency pair
problems is much higher than the complexity of termination itself.
The same holds for the property $\SNi$ of termination in infinitary
rewriting \cite{KV05}.
 We sketch how by the same argument $\SNi$ can
be concluded to be $\cpi{1}{1}$-complete.

A variant of dependency pair problems are dependency pair problems with minimality flag. We will show that for this
variant the complexity is back to that of termination: it is $\cpi{0}{2}$-complete.

\section{Preliminaries}

\subsection*{Term rewriting}

A \emph{signature $\asig$} is a finite set of symbols 
each having a fixed \emph{arity $\arity{f} \in \nat$}.
Let $\asig$ be a signature and $\avars$ a set of variable symbols such that $\asig \cap \avars = \setemp$.
The \emph{set $\ter{\asig,\avars}$ of terms over $\asig$ and $\avars$} is the smallest set satisfying:
\begin{itemize}
\item $\avars \subseteq \ter{\asig,\avars}$, and
\item $f(t_1,\dots,t_n) \in \ter{\asig,\avars}$ if $f \in \asig$ with arity $n$ and $\forall i: t_i \in \ter{\asig,\avars}$.
\end{itemize}
We use $x,y,z,\ldots$ to range over variables.
We frequently drop $\avars$ and write $\ter{\asig}$ for the set of terms over $\asig$
and a fixed, countably infinite set of variables $\avars$.
The set of positions $\pos{t} \subseteq \nat^*$ of a term $t \in \ter{\asig,\avars}$
is inductively defined by:
$\pos{\funap{f}{t_1,\ldots,t_n}} = 
 \{\lstemp\} \cup \{\lstconcat{i}{\apos} \where 1 \le i \le \arity{f},\, \apos \in \pos{t_i}\}$,
and
$\pos{x} = \{\lstemp\}$ for variables $x \in \avars$.
We use $\equiv$ for syntactical equivalence of terms.

A substitution $\asubst$ is a map $\asubst : \avars \to \ter{\asig,\avars}$ from variables to terms.
For terms $t \in \ter{\asig,\avars}$ and substitutions $\asubst$
we define $\subst{\asubst}{t}$ as the result of replacing each $x \in \avars$
in $t$ by $\funap{\asubst}{x}$.
That is, $\subst{\asubst}{t}$ is inductively defined by
$\subst{\asubst}{x} \defdby \funap{\asubst}{x}$ for variables $x \in \avars$
and otherwise
$\subst{\asubst}{\funap{f}{t_1,\ldots,t_n}} \defdby \funap{f}{\subst{\asubst}{t_1},\ldots,\subst{\asubst}{t_n}}$.
Let $\contexthole$ be a fresh %
symbol, $\contexthole \not\in \asig\join\avars$.
A \emph{context} $\acontext$ is a term from $\ter{\asig,\avars\join\{\contexthole\}}$
containing precisely one occurrence of $\contexthole$.
Then $\contextfill{\acontext}{s}$ denotes the term $\subst{\asubst}{\acontext}$
where $\funap{\asubst}{\contexthole} = s$ and $\funap{\asubst}{x} = x$ for all $x \in \avars$.

A \emph{term rewriting system (TRS)} over $\asig$, $\avars$ is a set $R$
pairs $\pair{\ell}{r} \in \ter{\asig,\avars}$,
called \emph{rewrite rules} and usually written as $\ell \to r$,
for which 
the \emph{left-hand side} $\ell$ is not a variable $\ell \not\in \avars$
and all variables in the \emph{right-hand side} $r$ occur in $\ell$, $\vars{r} \subseteq \vars{\ell}$. 
Let $\atrs$ be a TRS.
For terms $s, t \in \ter{\asig,\avars}$ we write $s \to_{\atrs} t$ 
if there exists a rule $\ell \to r \in \atrs$, a substitution $\asubst$
and a context $\acontext \in \ter{\asig,\avars\join\{\contexthole\}}$
such that
$s \equiv \contextfill{\acontext}{\subst{\asubst}{\ell}}$
and
$t \equiv \contextfill{\acontext}{\subst{\asubst}{r}}$;
$\to_{\atrs}$ is the \emph{rewrite relation} induced by $\atrs$.

\begin{definition}\normalfont\label{def:sn}
  Let $\atrs$ be a TRS 
  and $t \in \ter{\asig,\avars}$ a term.
  Then $\atrs$ is called
  \begin{itemize}
  \item 
  \emph{strongly normalizing (or terminating) on $t$},
  denoted $\SNrs{\atrs}{t}$,\\
  if every rewrite sequence starting from $t$ is finite.
  \item 
  \emph{weakly normalizing on $t$},
  denoted $\WNrs{\atrs}{t}$,\\
  if $t$ admits a rewrite sequence $t \mred s$ to a normal form $s$.
  \item
  \emph{confluent (or Church-Rosser) on $t$},
  denoted $\CRrs{\atrs}{t}$,\\
  if every pair of finite coinitial reductions starting from $t$ can be extended to a common reduct,
  that is, $\myall{t_1,t_2 \in \ter{\asig}}{t_1 \mredi t \mred t_2 \implies \myex{d}{t_1 \mred d \mredi t_2}}$.
  \item
  \emph{weakly confluent (or weakly Church-Rosser) on $t$}, 
  denoted $\WCRrs{\atrs}{t}$,\\
  if every pair of coinitial rewrite steps starting from $t$ can be joined,
  that is, $\myall{t_1,t_2 \in \ter{\asig}}{t_1 \redi t \red t_2 \implies \myex{d}{t_1 \mred d \mredi t_2}}$.
  \end{itemize}
  The TRS $\atrs$ is
  \emph{strongly normalizing} ($\SNr{\atrs}$),
  \emph{weakly normalizing} ($\WNr{\atrs}$),
  \emph{confluent} ($\CRr{\atrs}$)
  or \emph{weakly confluent} ($\WCRr{\atrs}$)
  if the respective property holds on all terms $t \in \ter{\asig,\avars}$.
  We say that $\atrs$
  is \emph{ground confluent}
  (or \emph{ground weakly confluent})
  if $\atrs$ is confluent (or weakly confluent) on all ground terms $t \in \ter{\asig,\setemp}$.
\end{definition}

\subsection*{Turing machines}

\begin{definition}\normalfont\label{def:tm}
  A \emph{Turing machine} $\tm$ is a quadruple $\quadruple{\tmstates}{\tmsig}{\tmstart}{\stmtrans}$
  consisting of:
  \begin{itemize}
  \item finite set of states $\tmstates$,
  \item an initial state $q_0 \in \tmstates$,
  \item a finite alphabet $\tmsig$ containing a designated symbol $\stmblank$, called \emph{blank}, and
  \item a partial \emph{transition function} 
        $\stmtrans \funin \autstates \times \tmsig \to \tmstates \times \tmsig \times \{\tmL,\tmR\}$.
  \end{itemize}
  A \emph{configuration} of a Turing machine is a pair $\pair{\astate}{\stmtape}$
  consisting of a state $\astate \in \tmstates$
  and the tape content $\stmtape \funin \zz \to \tmsig$
  such that the carrier $\{n \in \zz \where \tmtape{n} \ne \stmblank\}$ is finite.
  The set of all configurations is denoted $\tmconf{\tm}$.
  We define the relation $\tmstep$ on the set of configurations $\tmconf{\tm}$ as follows:
  $\pair{\astate}{\stmtape} \tmstep \pair{\astate'}{\stmtape'}$
  whenever:
  \begin{itemize}
   \item
     $\tmtrans{\astate}{\tmtape{0}} = \triple{\astate'}{f}{\tmL}$,
     $\funap{\stmtape'}{1} = f$ and $\myall{n \ne 0}{\funap{\stmtape'}{n+1} = \tmtape{n}}$, or
   \item
     $\tmtrans{\astate}{\tmtape{0}} = \triple{\astate'}{f}{\tmR}$,
     $\funap{\stmtape'}{-1} = f$ and $\myall{n \ne 0}{\funap{\stmtape'}{n-1} = \tmtape{n}}$.
  \end{itemize}
\end{definition}
Without loss of generality we assume that $\tmstates \cap \tmsig = \setemp$,
that is, the set of states and the alphabet are disjoint.
This enables us to denote configurations
as $\triple{w_1}{\astate}{w_2}$,
denoted $w_1^{-1} \astate w_2$ for short,
with $w_1,w_2 \in \tmsig^\infty$ and $\astate \in \tmstates$,
which is shorthand for $\pair{\astate}{\stmtape}$
where $\tmtape{n} = \funap{w_2}{n+1}$ for $0 \le n < \lstlength{w_2}$,
and $\tmtape{-n} = \funap{w_1}{n}$ for $1 \le n \le \lstlength{w_1}$
and $\tmtape{n} = \stmblank$ for all other positions $n \in \zz$.

The Turing machines we consider are deterministic.
As a consequence, final states are unique (if they exist),
which justifies the following definition.

\begin{definition}\normalfont
  Let $\tm$ be a Turing machine and $\pair{\astate}{\stmtape} \in \tmconf{\tm}$.
  We denote by $\tmfinal{\pair{\astate}{\stmtape}}{\tm}$
  the $\tmstep$-normal form of $\pair{\astate}{\stmtape}$
  if it exists and undefined, otherwise.
  Whenever $\tmfinal{\pair{\astate}{\stmtape}}{\tm}$ exists
  then we say that \emph{$\tm$ halts on $\pair{\astate}{\stmtape}$
  with final configuration $\tmfinal{\pair{\astate}{\stmtape}}{\tm}$}.
  Furthermore we say \emph{$\tm$ halts on $\stmtape$} as shorthand for
  \emph{$\tm$ halts on $\pair{\tmstart}{\stmtape}$}.
\end{definition}

Turing machines can compute $n$-ary functions $f \funin \nat^n \to \nat$
or relations $S \subseteq \nat^*$. 
We need only unary functions $\stmfun{\tm}$ and binary ${\tmrel{\tm}} \subseteq \nat \times \nat$ relations.

\begin{definition}\label{def:funcrel}\normalfont
  Let $\tm = \quadruple{\tmstates}{\tmsig}{\tmstart}{\stmtrans}$ be a Turing machine with $\stmsucc,\stmzer \in \tmsig$.
  We define a partial function $\stmfun{\tm} \funin \nat \pto \nat$ for all $n \in \nat$ by:
  \begin{align*}
  \tmfun{\tm}{n}
  = \begin{cases}
      m & \text{if }
      \tmfinal{\tmstart \stmsucc^{n} \stmzer}{\tm}
      = \ldots \astate \stmsucc^m \stmzer \ldots\\
      \text{undefined} & \text{otherwise}
    \end{cases}
  \end{align*}
  and for $\tm$ total (i.e.\ $\tm$ halts on all tapes) we define the binary relation ${\tmrel{\tm}} \subseteq \nat \times \nat$ by:
  \begin{align*}
  n \tmrel{\tm} m \;\Longleftrightarrow\;
    \tmfinal{\stmzer \stmsucc^{n} \tmstart \stmsucc^{m} \stmzer}{\tm}
    = \ldots \astate \stmzer \ldots
  \punc.
  \end{align*}
\end{definition}
Note that, the set $\{\,\tmrel{\tm} \where \tm \text{ a Turing machine that halts on all tapes}\,\}$
is the set of recursive binary relations on $\nat$.

\subsection*{The arithmetic and analytical hierachy}

In the introduction we briefly mentioned the arithmetical and analytical hierarchy. We now summarize the main notions and results relevant for this paper. For details see a standard text on mathematical logic, e.g.\ \cite{shoenfield:1967} or \cite{hinman:1978}, which contains more technical results regarding these hierrarchies.

\begin{definition}\normalfont\label{def:membership}
  Let $A \subset \nat$. The \emph{set membership problem for $A$}
  is the problem of deciding for given $a \in \nat$ whether $a \in A$.
\end{definition}

\begin{definition}\normalfont\label{def:reduce}
  Let $A \subseteq \nat$ and $B \subseteq \nat$.
  Then \emph{$A$ can be many-one reduced to $B$}, notation $A\lem B$
  if there exists a total computable function $f \funin \nat \to \nat$
  such that $\myall{n \in \nat}{n \in A \Leftrightarrow \funap{f}{n} \in B}$.
\end{definition}

\begin{definition}\normalfont\label{def:hard}
  Let $B \subseteq \nat$ and $\mathcal{P} \subseteq \powerset{\nat}$.
  Then $B$ is called $\mathcal{P}$-\emph{hard} if every $A \in \mathcal{P}$
  can be reduced to $B$,
  and $B$ is $\mathcal{P}$-\emph{complete} whenever additionally $B \in \mathcal{P}$.
\end{definition}

So a problem $B$ is $\mathcal{P}$-\emph{hard} if every problem $A \in
\mathcal{P}$ can be reduced to $B$: To decide ``$n\in A$'' we only
have to decide ``$f(n) \in B$'', where $f$ is the total computable
function that reduces $A$ to $B$.

The classification results in the following sections employ the
following well-known lemma, which states that whenever a problem $A$
can be reduced via a computable function to a problem $B$, then $B$ is
at least as hard as $A$.

\begin{lemma}\label{lem:reduce}
  If $A$ can be reduced to $B$
  and $A$ is $\mathcal{P}$-\emph{hard}, then $B$ is $\mathcal{P}$-\emph{hard}.
  \qed
\end{lemma}

\begin{remark}\label{rem:encoding}
Finite lists of natural numbers can be encoded as natural numbers
using the well-known G\"odel encoding: $\langle n_1, \ldots,
n_k\rangle := p_1^{n_1+1}\cdot \ldots p_k^{n_k+1}$, where $p_1,
\ldots, p_k$ are the first $k$ prime numbers. For this encoding, the
length function ($\lth\langle n_1, \ldots, n_k\rangle = k$) and the
decoding function ($\lth\langle n_1, \ldots, n_k\rangle_i = n_i$ if
$1\leq i\leq n$) are computable and it is decidable if a number is the
code of a finite list $\Seq(n)$.
\end{remark}

Using the encoding of finite lists of natural numbers, we can encode
Turing machines, terms and finite term rewriting systems.  The
following, known as Kleene's $T$-predicate, is a well-known decidable
problem: $T(m,\langle x\rangle,u) := m$ {encodes a Turing
  Machine} $M$, $u$ {encodes the} {computation of} $M$ {on} $x$
{whose end result is} $(u)_{\lth(u)}$.

An example from term rewriting that we can encode as a problem on
natural numbers is (we leave the encoding of terms as numbers
implicit), $s \to_{\atrs} t:= \exists \ell \to r \in \atrs\, \exists
\asubst\,\exists \acontext \,(s \equiv
\contextfill{\acontext}{\subst{\asubst}{\ell}}\wedge t \equiv
\contextfill{\acontext}{\subst{\asubst}{r}})$.  As all these
quantifiers are bounded (amounting to a finite search), this is a
decidable problem. Note that the fact that the TRS is finite and thus
finitely branching is crucial here.

Undecidable problems can be divided into a hierarchy of increasing
complexity, the first part of which is known as the {\em arithmetical
  hierarchy}. An example is the problem whether $t$ reduces in finitely
many steps to $q$: $t\mred_{\atrs} q := \exists \langle s_1, \ldots,
s_n\rangle ( t= s_1 \to_{\atrs} \ldots \to_{\atrs} s_n = q)$. This
problem is undecidable in general and it resides in the class
$\csig{0}{1}$, which is the class of problems of the form $\exists x\in
\nat\, P(x,n)$ where $P(x,n)$ is a decidable problem. (We usually
suppress the domain behind the existential quantifier.) Due to the
encoding of a finite list of numbers into numbers, a sequence of $\exists$
can always be replaced by one.

Similar to $\csig{0}{1}$, we have the class $\cpi{0}{1}$, which is the
class of problems of the form $\forall x\in \nat\, P(x,n)$ with
$P(x,n)$ a decidable problem. If we continue this procedure, we obtain
the classes $\csig{0}{n}$ and $\cpi{}{n}$ for every $n\in\nat$.

\begin{definition}\normalfont\label{def:arithclasses}
  $\csig{0}{n}$ is the class of problems of the form\\
 $A(k) = \exists
  x_n \forall x_{n-1} \ldots P(x_1, \ldots , x_n, k)$ where $P$ is
  decidable. So, there is a sequence of $n$ alternating quantifiers in
  front of $P$. 
$\cpi{0}{n}$ is the class of problems of the form
  $A(k) = \forall x_n \exists x_{n-1} \ldots P(x_1, \ldots , x_n, k)$
  where $P$ is decidable.
$\Delta^0_n := \csig{0}{n} \cap \cpi{0}{n}$
\end{definition}

That this definition is useful is based on the following fact, for
which refer to \cite{rogers:1967,hinman:1978,shoenfield:1967} for a
proof and further details.

\begin{remark}\label{remark:arith}
Every formula in first order arithmetic is equivalent to a formula in
{\em prenex normal form}, i.e.\ a formula with all quantifiers on the
outside of the formula.\\
For every formula of the form $\exists n \exists m \phi$  there is an equivalent formula of the form $\exists p \phi'$, where $\phi'$ has the same quantifier structure as $\phi$. Similarly, for every formula of the form $\forall n \forall m \phi$  there is an equivalent formula of the form $\forall p \phi'$, where $\phi'$ has the same quantifier structure as $\phi$.
\end{remark}

The reason one writes $0$ as a superscript is that all quantifiers
range over ``the lowest type'' $\nat$; there are no quantifiers of
higher types, like $\nat\rightarrow\nat$.  
So every arithmetical problem
is in one of the
classes of Definition \ref{def:arithclasses}. A natural question is
whether all these classes are distinct. A fundamental result in
mathematical logic says that they are, see \cite{shoenfield:1967}, \cite{rogers:1967} or
\cite{hinman:1978}.

\begin{lemma} \label{lem.arithhierprops}
$\REC =\Delta^0_1$ and for all $n\in\nat$, $\Delta^0_n \subsetneq
  \csig{0}{n} \subsetneq \Delta^0_{n+1}$ and $\Delta^0_n \subsetneq
  \cpi{0}{n} \subsetneq \Delta^0_{n+1}$. For all $n\in\nat$ and all $A\subset \nat$, $A\in \csig{0}{n} \Leftrightarrow \overline{A}\in \cpi{0}{n}$.
\end{lemma}

The arithmetic hierarchy is usually depicted as in Figure
\ref{fig:arithhier}, where every arrow denotes a proper
inclusion. Schematically one usually writes $\exists \REC$ for
$\csig{0}{1}$, $\forall\exists\REC$ for $\cpi{0}{2}$, etc. All classes are
closed under bounded quantification: if $A(n) \Leftrightarrow \exists
y< t(n)\,P(n,y)$ and $P$ is decidable, then $A$ is decidable (and
similarly for other classes in the hierarchy). To put it more
succinctly: $\forall < \mathcal{P} = \mathcal{P}$ for all classes
$\mathcal{P}$ in the arithmetic hierarchy.

\begin{figure}
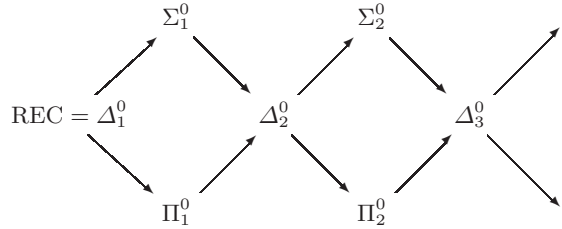

$$\begin{diagram}[height=2em,width=2em]
& & \csig{0}{1} & & & & \csig{0}{2} & & & &\\ 
& \NE & &    \SE & & \NE & & \SE & & \NE &\\ 
\REC = \Delta^0_1& & & &\Delta^0_2 & & & & \Delta^0_3&&\\ 
&\SE & & \NE & & \SE & & \NE & & \SE &\\ 
& & \cpi{0}{1} & & & & \cpi{0}{2} & & & &\\
\end{diagram}$$
\caption{Arithmetic Hierarchy\label{fig:arithhier}}
\end{figure}

To determine if a problem $A$ is essentially in a certain class
$\mathcal{P}$ (and not lower in the hierarchy), we first show that $A$
can be expressed with a formula of $\mathcal{P}$. This shows that $A$
is in $\mathcal{P}$ or lower. To prove that $A$ is not lower, we then
prove that $A$ is $\mathcal{P}$-complete.

Above the arithmetic hierrachy, we find the {\em analytic hierarchy},
where we also allow quantification over infinite sequences of
numbers. As variables ranging over infinite sequences we use $\alpha$,
$\beta$, etc. An example of an analtyical formula is $\forall \alpha
(\forall x(\alpha(x)\mred_{\atrs} \alpha(x+1)) \rightarrow \exists
x(\alpha(x) = \alpha(x+1)))$, stating that the rewrite system is
$\SN$. This is a $\cpi{1}{1}$-formula. In Section \ref{sec:termination} we will see that we can express $\SN$ for TRSs with a formula that is much lower in the hierarchy: it is $\cpi{0}{2}$. The proof essentially uses the fact that TRSs are finitely branching.

The class $\cpi{1}{1}$ is the class of problems of the form $\forall
\alpha\, \exists x\, P(n,\alpha,x)$, where $P$ decidable. Similarly
$\csig{1}{1}$ is the class of problems of the form $\exists \alpha\,
\forall x\, P(n,\alpha,x)$, where $P$ is decidable.  For analytical
problems we also have all kinds of simplification procedures
(analoguous to the ones of Remark \ref{remark:arith}).

\begin{lemma} In the analytical hierarchy we have the following ways of simplifying a sequence of quantifiers:
$$
{\forallt \forallt} \mapsto {\forallt}\:\: 
{\forall} \mapsto {\forallt}\:\:
{\exists \forallt}\mapsto {\forallt\exists}\:\:
{\forall\existst} \mapsto {\existst\forall}$$
\end{lemma}

For the first two simplifications, we of course have the analogous
versions with $\exists$. For the proof we refer to the standard
literature; here we just give a rough idea.  The meaning of the
first simplification is that a formula $\forallt \alpha \forallt \beta
\,\phi(\alpha,\beta)$ is equivalent to a formula of the form $\forall
\gamma \psi(\gamma)$, with $\psi$ in the same class as $\phi$. (Just
take $\psi(\gamma) := \phi((\gamma)_1, \gamma)_2))$, where
$(\gamma)_1$ denotes the sequence with $(\gamma)_1(n) =
(\gamma(n))_1$.) The meaning of the other simplifications should be
clear and from these simplifications one derives that each analytic
formula is equivalent to one of the form $Q_n \alpha_n Q_{n-1}
\alpha_{n-1} \ldots Q_0 x\,P(\alpha_1, \ldots , \alpha_n, x, k)$ where
$P$ is decidable and $\vec{Q}$ is a sequence of alternating
quantifiers. Any analytical problem can be written in this form.

\begin{definition}\normalfont\label{def:analyclasses}
  The analytical problems are the ones of the form\\ 
$Q_n
  \alpha_n Q_{n-1} \alpha_{n-1} \ldots Q_0 x\,P(\alpha_1, \ldots , \alpha_n, x, k)$ with $P$ is
  decidable and $\vec{Q}$ is a sequence of alternating quantifiers. If $n>0$ and $Q_n = \existst$, then it is in the class $\csig{1}{1}$. If $n>0$ and $Q_n = \forallt$, then it is in the class $\cpi{1}{1}$.
$\Delta^1_n := \Sigma^1_n \cap \Pi^1_n$
\end{definition}

For the analytical hierarchy we can draw a same diagram as the one in
Figure~\ref{fig:arithhier}: replace $\csig{0}{1}$ by $\csig{1}{1}$
etc. We have the same results as Lemma \ref{lem.arithhierprops}: each
class is a proper subclass of the ones above it. The whole arithmetic
hierarchy is also a proper subclass of the lowest class, $\Delta^1_1$.

\begin{lemma}\label{lem:complete}
  We have the following well-known
  results:
  \begin{enumerate}
    \item the \emph{special halting problem} $\{\,\tm \where \tm \text{ halts on the blank tape}\,\}$ is $\csig{0}{1}$-complete,
    \item the \emph{general halting problem} $\{\,\tm \where \tm \text{ halts on all inputs}\,\}$ is $\cpi{0}{2}$-complete,
    \item the \emph{totality problem} $\{\,\tm \where \text{$\tm$ halts on $\tmstart \stmsucc^n$ for every $n \in \nat$}\,\}$ is $\cpi{0}{2}$-complete,
    \item the set $\WF \defdby \{\,\tm \where {\tmrel{\tm}} \text{ is well-founded}\,\}$ is $\cpi{1}{1}$-complete.
  \end{enumerate}
\end{lemma}
These sets will be the basis for the hardness results in the following
sections: we will show that $\{\,\tm \where \tm \text{ halts on the
  blank tape}\,\}$ is many-one reducible to ``$\WN$ for a single
term'' and thus conclude that ``$\WN$ for a single term'' is
$\csig{0}{1}$. This will be done by effectively giving for every
Turing machine $\tm$, a TRS $R_\tm$ and a term $t_\tm$ such that
$$\tm \text{ halts on the blank tape \ \ iff\ \  }\WN_{R_\tm}(t_\tm)$$
Similar constructions will be carried out for the other problems that we consider.

\section{Strong and Weak Normalization} \label{sec:termination}
We use the translation of Turing machines $\tm$ to TRSs $\tmtrs{\tm}$ from \cite{jw:handbook}.
\begin{definition}\normalfont\label{def:tmtrs}
  For every Turing machine $\tm = \quadruple{\tmstates}{\tmsig}{\tmstart}{\stmtrans}$
  we define a TRS $\tmtrs{\tm}$ as follows.
  The signature is $\asig = \tmstates \cup \tmsig \cup \{\tmiblank\}$
  where the symbols $\astate \in \tmstates$ have arity 2,
  the symbols $f \in \tmsig$ have arity 1
  and $\tmiblank$ is a constant symbol, 
  which represents an infinite number of blank symbols.
  The rewrite rules of $\tmtrs{\tm}$ are:
  \begin{align*}
    \bfunap{\astate}{x}{\funap{f}{y}} &\to \bfunap{\astate}{\funap{f'}{x}}{y}
    &&\text{ for every }\tmtrans{\astate}{f} = \triple{\astate'}{f'}{\tmR}\\
    \bfunap{\astate}{\funap{g}{x}}{\funap{f}{y}} &\to \bfunap{\astate}{x}{\funap{g}{\funap{f'}{y}}}
    &&\text{ for every }\tmtrans{\astate}{f} = \triple{\astate'}{f'}{\tmL}
  \end{align*}
  together with four rules for `extending the tape':
  \begin{align*}
    \bfunap{\astate}{\tmiblank}{\funap{f}{y}} &\to \bfunap{\astate}{\tmiblank}{\tmblank{\funap{f'}{y}}}
    &&\text{ for every }\tmtrans{\astate}{f} = \triple{\astate'}{f'}{\tmL}\\
    \bfunap{\astate}{x}{\tmiblank} &\to \bfunap{\astate}{\funap{f'}{x}}{\tmiblank}
    &&\text{ for every }\tmtrans{\astate}{\stmblank} = \triple{\astate'}{f'}{\tmR}\\
    \bfunap{\astate}{\funap{g}{x}}{\tmiblank} &\to \bfunap{\astate}{x}{\funap{g}{\funap{f'}{\tmiblank}}}
    &&\text{ for every }\tmtrans{\astate}{\stmblank} = \triple{\astate'}{f'}{\tmL}\\
    \bfunap{\astate}{\tmiblank}{\tmiblank} &\to \bfunap{\astate}{\tmiblank}{\tmblank{\funap{f'}{\tmiblank}}}
    &&\text{ for every }\tmtrans{\astate}{\stmblank} = \triple{\astate'}{f'}{\tmL}
    \punc.
  \end{align*}
\end{definition}

We introduce a mapping from terms to configurations
to make the connection between 
the $\tm$ and the TRS $\tmtrs{\tm}$ precise.

\begin{definition}\normalfont\label{def:tmtrsmap}
  We define a mapping $\stmmap \funin \ter{\tmsig \cup \{\tmiblank\},\setemp} \pto \tmsig^*$ by:
  \begin{align*}
  \tmmap{\tmiblank} &\defdby \lstemp &
  \tmmap{\funap{f}{t}} \defdby \lstcat{s}{\tmmap{t}}
  \end{align*}
  for every $f \in \tmsig$ and $t \in \ter{\tmsig \cup \{\tmiblank\},\setemp}$.
  We define the set (intended) terms:
  \[\tmterms{\tm} \defdby \{\funap{\astate}{s,t} \where \astate \in \tmstates,\;s,t \in \ter{\tmsig \cup \{\tmiblank\},\setemp}\}
    \punc.\]
  Then we define a map $\stmint \funin \tmterms{\tm} \to \tmconf{\tm}$ by:
  \begin{align*}
    \tmint{\funap{\astate}{s,t}} \defdby \tmmap{s}^{-1}\astate\tmmap{t} \in \tmconf{\tm}
    \punc.
  \end{align*}
\end{definition}

\begin{lemma}\label{lem:tmtrs}
  Let $\tm$ be a Turing machine. Then $\tmtrs{\tm}$ simulates $\tm$, that is:
  \begin{enumerate}
    \item $\myall{c \in \tmconf{\tm}}{\funap{\stmint^{-1}}{c} \ne \setemp}$,
    \item \label{tmtrs:ii}
          for all terms $s \in \tmterms{\tm}$: $s \to_{\tmtrs{\tm}} t$ 
          implies $t \in \tmterms{\tm}$ and $\tmint{s} \tmstep \tmint{t}$, and
    \item for all terms $s \in \tmterms{\tm}$: 
          whenever $\tmint{s} \tmstep c$ then $\myex{t \in \funap{\stmint^{-1}}{c}}{s \to_{\tmtrs{\tm}} t}$.
  \end{enumerate}
\end{lemma}

The following is an easy corollary.

\begin{corollary}\label{cor:tmtrssn}
  For all $s \in \tmterms{\tm}$: $\SNrs{\tmtrs{\tm}}{s} \Longleftrightarrow \tm$ halts on $\tmint{s}$.
\end{corollary}
\begin{proof}
  Induction on item~\ref{tmtrs:ii} of Lemma~\ref{lem:tmtrs}.
\end{proof}

Let us elaborate a bit on Turing machines and the encoding of term rewriting.
\begin{remark}\label{rem:turingrew}
  As discussed in Remark~\ref{rem:encoding},
  terms and term rewriting systems can be encoded as natural numbers.
  Finite rewrite sequnences $\sigma \funin t_1 \to \ldots \to t_n$ can be encoded as lists of terms.
  Then of course a Turing machine can compute 
  the length of $\lstlength{\sigma} \defdby n$ of the sequence,
  every term $t_1$,\ldots,$t_n$,
  in particular the first $\first{\sigma} \defdby t_1$ and the last term $\last{\sigma} \defdby t_n$.
  Given the TRS as input,
  a Turing machine can check whether a natural number $n$
  corresponds to a valid rewrite sequence, that is, check $t_i \to t_{i+1}$ for every $i = 1,\ldots,(n-1)$.
  Furthermore for a given term $t$ and $n \in \nat$ it can calculate
  the set of all reductions of length $\le n$ admitted by $t$
  and thereby check properties like `all reductions starting from $t$ have length $\le n$'
  or `t is a normal form'.
\end{remark}

We arrive at our first results.
\begin{theorem}\label{thm:snsingle}
  The properties $\SN$ and $\WN$ for single terms are $\csig{0}{1}$-complete.
\end{theorem}
\begin{proof}
  For $\csig{0}{1}$-hardness we reduce the special halting problem to a termination problem for single terms.
  Therefore let $\tm$ be an arbitrary Turing machine.
  Then
  $\SNrs{\tmtrs{\tm}}{\bfunap{\tmstart}{\tmiblank}{\tmiblank}}$ if and only if
  $\tm$ halts on the blank tape
  by Corollary~\ref{cor:tmtrssn}.
  Moreover observe that $\tmtrs{\tm}$ is orthogonal and non-erasing, thus
  the $\SN$ and $\WN$ coincide \cite{terese:03}.
  Hence both properties $\SN$ and $\WN$ for single terms are $\csig{0}{1}$-hard by Lemma~\ref{lem:reduce}.

  To show that $\SN$ is in $\csig{0}{1}$, let $\atrs$ be a TRS and $t \in \ter{\asig,\avars}$ a term.
  Since $\atrs$ is finite, 
  $t$ is terminating if and only if
  there exists a bound on the length of the reductions admitted by $t$,
  that is, the following formula holds:
  \begin{gather*}
    \SNrs{\atrs}{t} \Longleftrightarrow \myex{n \in \nat}{\text{all reductions starting from $t$ have length $\le n$}}
  \end{gather*}
  Thus we have one existential number quantifier and
  by Remark~\ref{rem:turingrew} the predicate behind the quantifier is recursive.
  Hence $\SN$ for single terms is $\csig{0}{1}$-complete.

  To show that $\WN$ is in $\csig{0}{1}$, let $\atrs$ be a TRS and $t \in \ter{\asig,\avars}$ a term.
  The term $t$ is $\WN$ if there exists a reduction to a normal form:
  \begin{align*}
    \WNrs{\atrs}{t} \Longleftrightarrow \myex{r \in \nat}{
      &(\text{$r$ is a reduction}) \\ &\text{and } t \equiv \first{r} \text{ and } (\last{r} \text{ is a normal form})
    }
  \end{align*}
  This is a $\csig{0}{1}$-formula, hence $\WN$ for single terms is $\csig{0}{1}$-complete. \qed
\end{proof}

For showing $\cpi{0}{2}$-completeness of the uniform properties $\SN$
and $\WN$ we would like to use the equivalence
``$\SN(\tmtrs{\tm}) \Longleftrightarrow \tm$ halts on all inputs'',
in combinatio with Lemma~\ref{lem:complete} (ii). However, this does not work because of the following two problems:
\begin{enumerate}[(1)]
 \item \label{problem:one}
  In
  $\tmtrs{\tm}$ we have terms of the form $\astate(w,v)$, where
  $\astate$ is not the start state and $wv$ is some arbitrary (finite)
  tape content. That $\tm$ halts on all inputs, does not guarantee that
  $\tm$ halts when started in configuration
  $\pair{\astate}{wv}$.
  \item \label{problem:two}
  In
  $\tmtrs{\tm}$ we have terms of the form
  $\astate(\astate(w,v),u)$ that do not correspond to a configuration at
  all.
\end{enumerate}
To deal with problem~\eqref{problem:one}, we can use {\em type
introduction\/} \cite{Zantema:1994,terese:03}, since $\tmtrs{\tm}$
is non-collapsing. We asssign sort $s_0 \to s_0$ to every
$f \in \tmsig$, sort $s_0$ to $\tmiblank$ and sort $s_0\times
s_0 \to s_1$ to every $\astate \in \tmstates$.  The terms of sort
$s_0$ are normal forms.  The (non-variable) terms of sort $s_1$ are
in $\tmterms{\tm}$ after replacing all variables by $\tmiblank$, and
by Corollary~\ref{cor:tmtrssn} for all terms $t \in \tmterms{\tm}$
we have $\SNrs{\tmtrs{\tm}}{t}$ if and only if $\tm$ halts on
$\tmint{t}$.  Hence $\SNr{\tmtrs{\tm}}$ holds if and only if $\tm$
halts on all configurations $\tmconf{\tm}$.

We now need to deal with problem (1); we would like that $\tm$ halts
on all configurations $\tmconf{\tm}$ if and only if $\tm$ halts on all
inputs, starting from the initial state, but that's just not true. We need a Lemma about Turing machines; we use the following result by \cite{Herman:1971}.

\begin{lemma}[\cite{Herman:1971}]\label{lem:tmall}
For every Turing machine $\tm$ that computes a function $f:\nat\to\nat$ we can 
effectively construct a Turing machine $\widehat{\tm}$ such that
\begin{enumerate}
\item
$\widehat{\tm}$ also computes $f$,
\item
$\tm$ hals on all configurations if and only if $f$ is total
\end{enumerate}
\end{lemma}

So, if $\tm$ halts on all inputs (when started in the initial state),
then $\widehat{\tm}$ halts on all configurations. This solves problem
(1) and we have the following Corollary, which follows from the fact
that the general halting problem (set (ii) in Lemma \ref{lem:complete}) many-one reduces to the universal halting problem (the set in the
Corollary), using Lemma \ref{lem:reduce}. Basically, this corollary has already been stated and proved
in \cite{Herman:1971}.

\begin{corollary}\label{cor:tmall}
  The \emph{uniform halting problem} 
  $$\{\,\tm \where \tm \text{ halts on all configurations $\pair{\astate}{\stmtape} \in \tmconf{\tm}$}\,\}$$
  is $\cpi{0}{2}$-complete.
\end{corollary}

\begin{theorem}\label{tm:snuniform}
  The properties uniform $\SN$ and $\WN$ are $\cpi{0}{2}$-complete.
\end{theorem}
\begin{proof}
  For $\cpi{0}{2}$-hardness: we have seen how the universal halting problem for $\tm$ many-one reduces to the uniform termination problem for $\tmtrs{\tm}$.
Since $\tmtrs{\tm}$ is orthogonal and non-erasing
  $\SN$ and $\WN$ coincide \cite{terese:03}.
  Hence $\SN$ and $\WN$ are both $\cpi{0}{2}$-hard by Lemma~\ref{lem:reduce}.
  That the uniform properties $\SN$ and $\WN$ are in $\cpi{0}{2}$ 
  follows from the fact that these properties for single terms 
  can be described by $\csig{0}{1}$-formulas
  and the uniform property `adds' a universal number quantifier.
\end{proof}

\section{Confluence and Ground Confluence} \label{sec:confluence}

We investigate the complexity of confluence ($\CR$) and ground confluence ($\gCR$)
both uniform and for single terms.

For proving $\cpi{0}{2}$-completeness of confluence
one would like to use an extension of $\tmtrs{\tm}$
with the following rules:
\begin{align*}
    \bfunap{\msf{run}}{x}{y} &\to \tmT\\
    \bfunap{\msf{run}}{x}{y} &\to \bfunap{\tmstart}{x}{y}\\
    \bfunap{\astate}{x}{\funap{f}{y}} &\to \tmT
    &&\text{ for every $f \in \tmsig$ with }\tmtrans{\astate}{f}\text{ is undefined}
\end{align*}
On first glance it seems that $\bfunap{\tmstart}{s}{t} \to^* \tmT$
if the Turing machine $\tm$ halts on all configurations.
However, a problem arises if $s$ and $t$ contain variables;
e.g. if $s$ or $t$ are variables themselves.
We solve the problem as follows.
For Turing machines $\tm$
we define the TRS $\btrs_\tm$ to consist of the rules of the TRS $\tmtrs{\tm}$ extended by:
\begin{align}
  \bfunap{\msf{run}}{x}{\tmiblank} &\to \tmT \label{rule:t}\\
  \bfunap{\msf{run}}{\tmiblank}{y} &\to \bfunap{\tmstart}{\tmiblank}{y} \label{rule:tm}\\
  \bfunap{\astate}{x}{\funap{f}{y}} &\to \tmT
  &&\text{ for every $f \in \tmsig$ with }\tmtrans{\astate}{f}\text{ is undefined} \label{rule:tmterm}\\
  \bfunap{\msf{run}}{x}{\tmsucc{y}} &\to \bfunap{\msf{run}}{\tmsucc{x}}{y} \label{rule:convr}\\
  \bfunap{\msf{run}}{\tmsucc{x}}{y} &\to \bfunap{\msf{run}}{x}{\tmsucc{y}} \label{rule:convl}
  \punc.
\end{align}
Then $\tmT$ and $\bfunap{\tmstart}{\tmiblank}{s}$ are convertible using the rules~\eqref{rule:t}--\eqref{rule:convl}
if and only if $s$ is a ground term of the form $\funap{\stmsucc^n}{\tmiblank}$.

\begin{theorem}\label{thm:cr}
  Uniform confluence ($\CR$), and uniform ground confluence ($\gCR$)
  are $\cpi{0}{2}$-complete.
\end{theorem}

\begin{proof}
  For proving $\cpi{0}{2}$-hardness we reduce the totality problem to confluence.
  Let $\tm$ be an arbitrary Turing machine.
  We consider the TRS $\btrs_\tm$ defined above.
  We employ type introduction~\cite{Aoto97}:
  we assign
  sort $s_0$ to $\tmsig \cup \{\tmiblank\}$
  and sort $s_1$ to every symbol in $\{\msf{run},\,\tmT\} \cup \tmstates$;
  the obtained many-sorted TRS is confluend if and only if $\btrs$ is.
  Note that the terms of sort $s_0$ are normal forms
  and for terms of $s_1$ with root symbol $\neq$ `$\msf{run}$' the reduction is deterministic
  (exhibits no branching).
  Therefore it suffices to consider the case
  $$s_2 \redi_{\eqref{rule:tm}} s_1 \redi_{\eqref{rule:convr}}^* 
   \funap{\msf{run}}{t_1,t_2} \red_{\eqref{rule:convl}}^* s_3 \red_{\eqref{rule:t}} \tmT$$
  where $t_1,t_2 \in \ter{\tmsig \cup \{\tmiblank\},\avars}$.
  From the existence of such rewrite sequences we conclude that
  there exists $n \in \nat$ such that $s_1 \equiv \bfunap{\msf{run}}{\tmiblank}{\funap{\stmsucc^n}{\tmiblank}}$,
  $s_3 \equiv \bfunap{\msf{run}}{\funap{\stmsucc^n}{\tmiblank}}{\tmiblank}$, and
  $s_2 \equiv \bfunap{\tmstart}{\tmiblank}{\funap{\stmsucc^n}{\tmiblank}}$.
  On the other hand for every $n \in \nat$ such rewrite sequences exist.
  As a consequence the TRS $\btrs$ is confluent if and only if 
  $\bfunap{\tmstart}{\tmiblank}{\funap{\stmsucc^n}{\tmiblank}} \red_{\btrs}^* \tmT$
  for every $n \in \nat$, that is,
  if and only if $\tm$ halts on $\tmstart \stmsucc^n$ for every $n \in \nat$.
  Moreover since the only critical terms $\funap{\msf{run}}{t_1,t_2}$ are ground terms,
  we conclude that ground confluence coincides with confluence for $\btrs$.
  Hence we have shown $\cpi{0}{2}$-hardness.

  To show that both properties are in $\csig{0}{2}$ let $\atrs$ be a TRS.
  Then $\atrs$ is confluent if and only if the following formula holds:
  \begin{align*}
    \CRr{\atrs} \Longleftrightarrow\ &\myall{t \in \nat}{\myall{r_1,r_2 \in \nat}{\myex{r_1',r_2' \in \nat}{\\
      (&((\text{$t$ is a term}) \text{ and  } (\text{$r_1$, $r_2$ are reductions})
       \text{ and } t \equiv \first{r_1} \equiv \first{r_2})\\
      &
       \begin{aligned}
       \implies\ &((\text{$r_1'$ and $r_2'$ are reductions})\\
       &\text{and } (\last{r_1} \equiv \first{r_1'}) \text{ and } (\last{r_2} \equiv \first{r_2'})\\
       &\text{and } (\last{r_1'} \equiv \last{r_2'})))
       \end{aligned}
    }}}\punc.
  \end{align*}
  By quantifier compression we can simplify the formula such that
  there is only universal followed by an existencial quantifier.
  Note that a formula for $\gCR$ is obtained
  by relacing `$t$ is a term' by `$t$ is a ground term'.
  Therefore both the uniform properties $\CR$ and $\gCR$ are $\cpi{0}{2}$-complete.
\end{proof}

\begin{theorem}\label{thm:crsingle}
  Confluence ($\CR$), and ground confluence ($\gCR$)
  for single terms are $\cpi{0}{2}$-complete.
\end{theorem}

\begin{proof}
  For $\cpi{0}{2}$-hardness we use the totality problem.
  Let $\tm$ be an arbitrary Turing machine.
  We define the TRS $\btrs$ as $\tmtrs{\tm}$ extended by the following rules:
  \begin{gather*}
    \begin{aligned}
    \funap{\msf{run}}{x} &\to \tmT &\quad
    \funap{\msf{run}}{x} &\to \funap{\msf{run}}{\tmsucc{x}} &\quad
    \funap{\msf{run}}{x} &\to \bfunap{\tmstart}{\tmiblank}{x}
    \end{aligned}\\
    \bfunap{\astate}{x}{\funap{f}{y}} \to \tmT
    \quad\text{ for every $f \in \tmsig$ with }\tmtrans{\astate}{f}\text{ is undefined} \label{rule:tmterm}
    \punc.
  \end{gather*}
  The term $t \defdby \funap{\msf{run}}{\tmiblank}$
  rewrites to $\tmT$ and $\bfunap{\tmstart}{\tmiblank}{\funap{\stmsucc^n}{x}}$ for every $n \in \nat$.
  Furthermore we have $\bfunap{\tmstart}{\tmiblank}{\funap{\stmsucc^n}{\tmiblank}} \red_{\btrs}^* \tmT$
  if and only if $\tm$ halts on $\tmstart \stmsucc^n$.
  As a consequence $\CR$ and $\gCR$ for single terms are $\cpi{0}{2}$-hard.

  For $\cpi{0}{2}$-completeness note that we can formalize $\CR$ and $\gCR$
  for single terms simply by dropping the universal quantification over all terms ($\forall t \in \nat$)
  from the respecitive $\cpi{0}{2}$-formulas for the uniform properties in the proof of Theorem~\ref{thm:cr}.
\end{proof}

\section{Weak Confluence and Weak Ground Confluence}

We investigate the complexity of weak confluence ($\WCR$) and weak ground confluence ($\gWCR$)
both uniform and for single terms.

\begin{theorem}\label{thm:wcr}
  The properties weak confluence ($\CR$) both for single terms and uniform,
  and weak ground confluence ($\gCR$) for single terms are $\csig{0}{1}$-complete.
\end{theorem}

\begin{proof}
  For $\csig{0}{1}$-hardness we use the special halting problem.
  Let $\tm$ be an arbitrary Turing machine.
  We define the TRS $\btrs$ to consist of the rules of $\tmtrs{\tm}$ extended by the following rules:
  \begin{gather*}
    \begin{aligned}
      \msf{run} &\to \tmT &\quad
      \msf{run} &\to \bfunap{\tmstart}{\tmiblank}{\tmiblank}
    \end{aligned}\\
    \bfunap{\astate}{x}{\funap{f}{y}} \to \tmT
    \quad\text{ for every $f \in \tmsig$ with }\tmtrans{\astate}{f}\text{ is undefined} \label{rule:tmterm}
    \punc.
  \end{gather*}
  The only critical pair is
  $\tmT \redi \msf{run} \to \bfunap{\tmstart}{\tmiblank}{\tmiblank}$,
  and we have $\bfunap{\tmstart}{\tmiblank}{\tmiblank} \red_{\btrs}^* \tmT$,
  if and only if $\tm$ halts on the blank tape.
  By the Critical Pairs Lemma~\cite{terese:03}
  we know that $\WCR$ holds if and only if all critical pairs are convergent (can be joined).
  Hence uniform $\WCR$, and for single terms $\WCR$ and $\gWCR$ ($t \defdby \msf{run}$) are $\csig{0}{1}$-hard.

  A Turing machine can compute 
  on the input of a TRS $\atrs$ all (finitely many) critical pairs,
  and on the input of a TRS $\atrs$ and a term $t$ all (finitely many) one step reducts of $t$.
  Therefore it suffices to show that the following problem is in $\csig{0}{1}$:
  decide on the input of a TRS $\btrs$, $n \in \nat$ and terms $t_1,s_1,\ldots,t_n,s_n$
  whether for every $i = 1,\ldots,n$ the terms $t_i$ and $s_i$ have a common reduct.
  This property can be described by the following $\csig{0}{1}$ formula:
  \begin{align*}
    \myex{r \in \nat}{
      (&(\text{$r$ is list $r_1,\ldots,r_{2\cdot n}$ of legnth $2 \cdot n$})\\
       &\text{and for $i = 1,\ldots,n$ we have}\\
       &\quad
        \begin{aligned}
          &(\text{$r_{2\cdot i}$, $r_{2\cdot i + 1}$ are reductions}) \text{ and } (\first{r_{2\cdot i}} \equiv t_i)\\
          &\text{and } (\first{r_{2\cdot i + 1}} \equiv s_i)
           \text{ and } (\last{r_{2\cdot i}} \equiv \last{r_{2\cdot i+1}})
          \punc.
        \end{aligned}
    }
  \end{align*}
\end{proof}

Surprisingly it turns out that uniform weak ground confluence is $\csig{0}{2}$-complete
and thereby harder than uniform weak confluence (for the set of all open terms).

\begin{theorem}\label{thm:wgcr}
  Uniform weak ground confluence ($\gCR$) is $\cpi{0}{2}$-complete.
\end{theorem}

\begin{proof}
  For $\cpi{0}{2}$-hardness we use the uniform halting problem.
  Let $\tm$ be a Turing machine.
  We define the TRS $\btrs$ as extension of $\tmtrs{\tm}$ with:
  \newcommand{\isNr}{\funap{\msf{isNr}}}
  \newcommand{\ok}{\funap{\msf{ok}}}
  \begin{align*}
    \bfunap{\msf{run}}{x}{y} &\to \tmT\\
    \bfunap{\msf{run}}{x}{y} &\to \bfunap{\tmstart}{x}{y}
    \punc,
  \end{align*}
  and rules 
  \begin{align*}
  \bfunap{\astate}{\funap{f}{\vec{x}}}{\funap{g}{\vec{y}}} &\to \tmT
  \end{align*}
  for all combinations of symbols of $f$, $g$
  such that the left-hand side is not matched by any of the rules in $\tmtrs{\tm}$.
  Here $\vec{x}$ and $\vec{y}$ are vectors of distinct variables such that the left-hand side of the rules are left-linear.

  Assume there exists a configuration $c$ on which $\tm$ does not halt.
  Then by Lemma~\ref{lem:tmtrs} there exists $\bfunap{\astate}{s}{t} \in \funap{\stmint^{-1}}{c}$
  and by Corollary~\ref{cor:tmtrssn} $\tmtrs{\tm}$ is not terminating on $\bfunap{\astate}{s}{t}$.
  Every reduct of $\bfunap{\astate}{s}{t}$ is an $\tmtrs{\tm}$-redex
  and contains no further redexes.
  In particular, none of the extended rules is applicable to any reduct.
  Hence $\bfunap{\astate}{s}{t} \not\mred \tmT$ and thus $\tmT \redi \bfunap{\msf{run}}{s}{t} \to \bfunap{\astate}{s}{t}$
  is not joinable.

  Assume that $\tm$ halts on all configurations.
  Let $D = \{\msf{run}\} \cup \tmstates$. 
  Let $V$ be the set of ground terms having a root symbol from $D$.
  All symbols apart from $D$ are constructor symbols.
  Hence for (weak) confluence it suffices to show that every reduct of a term in $V$ rewrites to $\tmT$.
  Every term from $V$ is a redex and all reducts of terms in $V$ are in $V \cup \{\tmT\}$.
  Thus it suffices to show that no term in $V$ admits an infinite root rewrite sequence.
  Such a sequence can only exists if a ground of the form $\bfunap{\astate}{s}{t}$
  admits an infinite $\tmtrs{\tm}$-root rewrite sequence.
  Below the root (which is in $\tmstates$) the rules from $\tmtrs{\tm}$
  match only symbols from $\tmsig \cup \{\tmiblank\}$.
  Let $s'$ (and $t'$) be obtained from $s$ (and $t$, respectively) by replacing
  all subterms having a root symbol not in $\tmsig \cup \{\tmiblank\}$ with $\tmiblank$.
  Then $\bfunap{\astate}{s'}{t'}$ admits an infinite $\tmtrs{\tm}$-rewrite sequence,
  $s', t' \in \ter{\asig \cup \{\tmiblank\},\setemp}$, and $\bfunap{\astate}{s'}{t'} \in \tmterms{\tm}$.
  Consequently $\tmint{\bfunap{\astate}{s'}{t'}}$
  is a non-terminating configuration of $\tm$ by Corollary~\ref{cor:tmtrssn},
  contradicting the assumption that $\tm$ halts on all configurations.
  \qed
\end{proof}

\section{Dependency Pair Problems}
In this section we present the remarkable result that finiteness of dependency pair problems, 
although invented for proving termination, is of a much higher level of complexity than termination
itself: it is $\cpi{1}{1}$-complete, both uniform and for single terms. This only holds for the 
basic version of dependency pairs; for the version with minimality flag we will show it is of 
the same level as termination itself. 

For relations $\to_1, \to_2$ we write $\to_1 / \to_2 \; = \; \to_2^* \cdot \to_1$.
For TRSs $R$, $S$ instead of $\SN(\redrroot{R}/\to_S)$ we shortly write $\SNdp{R}{S}$; in the 
literature \cite{GTS04} this is called {\em finiteness of the dependency pair problem} $\{R,S\}$. 
So $\SNdp{R}{S}$ means that every infinite $\redrroot{R} \cup \to_S$ reduction contains 
only finitely many $\redrroot{R}$ steps.  The 
motivation for studying this comes from the dependency pair approach \cite{AG00} for proving 
termination: for any TRS $R$ we can easily define a TRS $\DP(R)$ such that we have
\[ \SNdp{\DP(R)}{R} \Longleftrightarrow \SN(R). \]

The main result of this section is $\cpi{1}{1}$-completeness of $\SNdp{R}{S}$, 
even of $\SNdp{S}{S}$,
for both the uniform and the single term variant. In the next section we will consider the
variant $\SNdpm{R}{S}$ with minimality flag which only makes sense for the uniform variant, 
and show that it behaves like normal termination: it is $\cpi{0}{2}$-complete.

For proving $\cpi{1}{1}$-hardness of $\SNdp{S}{S}$
 we now adopt Definition~\ref{def:tmtrs}, the translation of Turing machines to TRSs.
The crucial difference is that every step of the Turing machine
`produces' one output pebble `$\stmpeb$',
thereby we achieve that the TRS $\tmtrspeb{\tm}$ is top-terminating
even if the Turing machine $\tm$ does not terminate.

\begin{definition}\normalfont\label{def:tmtrs-adap}
  For every Turing machine $\tm = \quadruple{\tmstates}{\tmsig}{\tmstart}{\stmtrans}$
  we define the TRS $\tmtrspeb{\tm}$ as follows.
  The signature $\asig = \tmstates \cup \tmsig \cup \{\tmiblank,\stmpeb,\tmT\}$ where $\stmpeb$ is a unary symbol,
  $\tmT$ is a constant symbol,
  and the rewrite rules of $\tmtrspeb{\tm}$ are:
  \begin{align*}
    \ell &\to \tmpeb{r}
    &&\text{ for every }\ell \to r \in \tmtrs{\tm}
  \end{align*}
  and rules for rewriting to $\tmT$ after successful termination:
  \begin{align*}
    \bfunap{\astate}{x}{\tmzer{y}} &\to \tmT
    &&\text{ whenever }\tmtrans{\astate}{\stmsucc}\text{ is undefined}\\
    \tmpeb{\tmT} &\to \tmT
    \punc.
  \end{align*}
\end{definition}

Then we obtain the following lemma. (Recall the Definition of $\tmrel{\tm}$ in \ref{def:funcrel}.)
\begin{lemma}
  For every Turing machine $\tm = \quadruple{\tmstates}{\tmsig}{\tmstart}{\stmtrans}$
  and $n,m \in \nat$ we have $n \tmrel{\tm} m$ if and only if
  $\funap{\tmstart}{\stmsucc^n,\stmsucc^m} \mred_{\tmtrspeb{\tm}} \tmT$.\qed
\end{lemma}

Moreover we define an auxiliary TRS $\trspickn$ for generating a random natural number $n \in \nat$
in the shape of a term  $\funap{\sfs^n}{\tmzer{\tmiblank}}$:

\begin{definition}\normalfont\label{def:pickn}
  We define the TRS $\trspickn$ to consist of the following three rules:
  \begin{align*}
    \pickn &\to \fc{\pickn} &
    \pickn &\to \picknok{\tmzer{\tmiblank}} &
    \fc{\picknok{x}} &\to \picknok{\fs{x}}
    \punc.
  \end{align*}
\end{definition}
\begin{lemma}\label{lem:pickn}
  The TRS $\trspickn$ has the following properties:
  \begin{itemize}
   \item $\pickn \mred \picknok{\funap{\sfs^n}{\tmzer{\tmiblank}}}$ for every $n \in \nat$, and
   \item whenever $\pickn \mred \picknok{t}$ for some term $t$ then
         $t \equiv \funap{\sfs^n}{\tmzer{\tmiblank}}$ for some $n \in \nat$.
  \end{itemize}
\end{lemma}

Now we are ready to prove $\cpi{1}{1}$-completeness
of dependency pair problems.
\begin{theorem}\label{thm:snilocal}
Both  $\SNdps{\atrs}{\btrs}{t}$ and  $\SNdp{\atrs}{\btrs}$ are $\cpi{1}{1}$-complete.
\end{theorem}

\begin{proof}
We prove $\cpi{1}{1}$-hardness even for the case where $\atrs$ and $\btrs$ coincide. We do this by using that the set $\WF$ is 
$\cpi{1}{1}$-complete, that is, checking well-foundedness of $\tmrel{\tm}$.
  Let $\tm$ be an arbitrary Turing machine. 
  From $\tm$ we construct a TRS $\btrs$ together with a term $t$ such that:
  \[\SNdp{\btrs}{\btrs} \Longleftrightarrow \SNdps{\btrs}{\btrs}{t} \Longleftrightarrow {\tmrel{\tm}} \text{ is well-founded}\punc.\]
  Let $\btrs$ consist of the rules of $\tmtrspeb{\tm} \uplus \trspickn$
  together with:
  \begin{align}
    \tfunap{\msf{run}}{\tmT}{\picknok{x}}{\picknok{y}} &\to \tfunap{\msf{run}}{\bfunap{\tmstart}{x}{y}}{\picknok{y}}{\pickn}
    \punc,\label{rule:dp}
  \end{align}
  and define $t \defdby \tfunap{\msf{run}}{\tmT}{\pickn}{\pickn}$.

As the implication from the first to the second item is trivial, we
only have to prove (1) $\SNdps{\btrs}{\btrs}{t} \Longleftrightarrow
{\tmrel{\tm}} \text{ is well-founded}$ and (2) ${\tmrel{\tm}} \text{
is well-founded} \Longleftrightarrow \SNdp{\btrs}{\btrs}$.

(1)  Suppose $\SNdps{\btrs}{\btrs}{t}$ and assume there is an infinite descending $\tmrel{\tm}$-sequence: $n_1 \tmrel{\tm} n_2 \tmrel{\tm} \ldots$. Then we have:
  \begin{align*}
    \tfunap{\msf{run}}{\tmT}{\pickn}{\pickn}
    \mred&\ \tag{$*$}\label{dpseq}
      \tfunap{\msf{run}}{\tmT}{\picknok{\funap{\sfs^{n_1}}{\tmzer{\tmiblank}}}}{\picknok{\funap{\sfs^{n_2}}{\tmzer{\tmiblank}}}}\\
    \redrroot{\btrs}&\  \tfunap{\msf{run}}
             {\bfunap{\tmstart}{\funap{\sfs^{n_1}}{\tmzer{\tmiblank}}}{\funap{\sfs^{n_2}}{\tmzer{\tmiblank}}}}
             {\picknok{\funap{\sfs^{n_2}}{\tmzer{\tmiblank}}}}{\pickn}\\
    \mred&\ \tfunap{\msf{run}}
             {\tmT}
             {\picknok{\funap{\sfs^{n_2}}{\tmzer{\tmiblank}}}}{\picknok{\funap{\sfs^{n_3}}{\tmzer{\tmiblank}}}}\\
    \redrroot{\btrs}&\ \ldots
  \end{align*}
Note that $\bfunap{\tmstart}{\funap{\sfs^{n_i}}{\tmzer{\tmiblank}}}{\funap{\sfs^{n_{i+1}}}{\tmzer{\tmiblank}}} \mred \tmT$
  (for all $i \ge 1$) because
  $\tm$ computes the binary predicate $\tmrel{\tm}$. So we have an infinite reduction starting from $t$, contradicting $\SNdps{\btrs}{\btrs}{t}$. So there is no infinite descending $\tmrel{\tm}$-sequence.

(2) Suppose that ${\tmrel{\tm}} \text{
is well-founded}$ and assume that $\sigma$ is a rewrite sequence containing infinitely many root steps.
  Note that \eqref{rule:dp} is the only candidate for a rule which can be applied infinitely often at the root.
  Hence all terms in $\sigma$ have the root symbol $\msf{run}$.
  We consider the first three applications of \eqref{rule:dp} at the root in $\sigma$.
  After the first application the third argument of $\msf{run}$ is $\pickn$.
  Therefore after the second application the second argument of $\msf{run}$ 
  is a reduct of $\pickn$ and the third is $\pickn$.
  Then before the third application the first argument is $\tmT$,
  and both the second and the third argument are reducts of $\pickn$. Thus $\SNdps{\btrs}{\btrs}{t}$ cannot hold.

It remains to prove that both $\SNdp{\atrs}{\btrs}$ and $\SNdps{\atrs}{\btrs}{t}$ are in $\cpi{1}{1}$. Let $\atrs$ 
and $\btrs$ be TRSs.  Then $\SNdp{\atrs}{\btrs}$ holds if and only if all
  $\redrroot{\atrs} \cup \redr{\btrs}$ reductions contain only a finite number of $\redrroot{\atrs}$ steps.
  An infinite reduction can be encoded as a function $\alpha \funin \nat \to \nat$
  where $\funap{\alpha}{n}$ is the $n$-th term of the sequence.
  We can express the property as follows:
  \begin{align*}
    \SNdp{\atrs}{\btrs} \Longleftrightarrow\ &\myall{\alpha \funin \nat \to \nat}{\\
      (&(\myall{n \in \nat}{\funap{\alpha}{n}\text{ rewrites to }\funap{\alpha}{n+1}\text{ via $\redrroot{\atrs} \cup \redr{\btrs}$}}) \implies\\
      &\myex{m_0 \in \nat}{\myall{m \ge m_0}{
         \neg (\funap{\alpha}{m}\text{ rewrites to }\funap{\alpha}{m+1}\text{ via $\redrroot{\btrs}$})
       }})
    }\punc,
  \end{align*}
  containing one universal function quantifier in front of an arithmetic formula.
  Here the predicate `$n$ rewrites to $m$' tacitly includes a check that
  both $n$ and $m$ indeed encode terms (which estabishes no problem for a Turing machine).
  For the property $\SNdps{\atrs}{\btrs}{t}$ we simply add the condition $t = \funap{f}{1}$ 
  to restrict the quantification to such rewrite sequences $f$ that start with $t$.
  Hence $\SNdp{\atrs}{\btrs}$ and $\SNdps{\atrs}{\btrs}{t}$ are $\cpi{1}{1}$-complete.
\qed
\end{proof}

We now sketch how this proof also implies  $\cpi{1}{1}$-completeness of the property $\SNi$ in infinitary rewriting, for its 
definition and basic observations see \cite{KV05}.
Since in Theorem \ref{thm:snilocal} we proved $\cpi{1}{1}$-hardness even for the case where $\atrs$ and $\btrs$ coincide,
we conclude that $\SNdp{\btrs}{\btrs}$ is $\cpi{1}{1}$-complete. This property $\SNdp{\btrs}{\btrs}$ states that every 
infinite $\btrs$-reduction contains only finitely many root steps. This is the same as the property $\SNw$ when restricting
to finite terms; for the definition of $\SNw$ see \cite{Z08} (basically, it states that in any infinite reduction the position of the contracted redex moves to infinity). However, when extending to infinite terms it still holds that 
for the TRS $S$ in the proof of Theorem \ref{thm:snilocal} the only infinite $S$-reduction containing infinitely many 
root steps is of the shape given in that proof, only consisting of finite terms. So $\SNw$ for all terms (finite and infinite)
is $\cpi{1}{1}$-complete. It is well-known that for left-linear TRSs the properties $\SNw$ and $\SNi$ coincide, see e.g. 
\cite{Z08}. Since the TRS $S$ used in the proof of Theorem \ref{thm:snilocal} is left-linear we conclude that the property
$\SNi$ for left-linear TRSs is $\cpi{1}{1}$-complete.

\section{Dependency Pair Problems with Minimality Flag}
A variant in the dependency pair approach is the dependency pair problem with minimality flag. 
Here in the infinite 
$\redrroot{R} \cup \to_S$ reductions all terms are assumed to be $S$-terminating. This can
be defined as follows. On the level of relations $\to_1, \to_2$ we write 
\[\to_1 /_{\msf{min}} \to_2 \; = \; (\to_2^* \cdot \to_1) \cap \to_{\SN(\to_2)},\]
where the relation $\to_{\SN(\to_2)}$ is defined to consist of all 
pairs $(x,y)$ for which $x$ is $\to_2$-terminating. 
For TRSs $R,S$ instead of $\SN(\redrroot{R}/_{\msf{min}}\to_S)$ we shortly write 
$\SNdpm{R}{S}$. In \cite{GTS04} this is called finiteness of the dependency pair problem
$(R,Q,S)$ with minimality flag; in our setting the middle TRS $Q$ is empty.
Again the motivation for this definition is in proving termination: from \cite{AG00} we 
know 
\[ \SNdpm{\DP(R)}{R} \Longleftrightarrow \SN(R). \]
For $\SNdpm{R}{S}$ it is not clear how to define a single term variant, in particular for
terms that are not $S$-terminating.
In this section we prove that $\SNdpm{R}{S}$ is $\cpi{0}{2}$-complete.
For doing so first we give some lemmas.

\begin{lemma}
\label{lemdp1}
Let $R,S$ be TRSs. Then $\SNdpm{R}{S}$ holds if and only if 
\[(\redrroot{R} \cup \to_S) \cap \to_{\SN(\to_S)}\]
is terminating.
\end{lemma}
\begin{proof}
By definition $\SNdpm{R}{S}$ is equivalent to termination of \linebreak
$(\to_S^* \cdot \redrroot{R}) \cap \to_{\SN(\to_S)}$. Since
\[(\to_S^* \cdot \redrroot{R}) \cap \to_{\SN(\to_S)} \;\; \subseteq \;\;
((\redrroot{R} \cup \to_S) \cap \to_{\SN(\to_S)})^+, \]
the `if'-part of the lemma follows.

For the `only if'-part assume  $(\redrroot{R} \cup \to_S) \cap \to_{\SN(\to_S)}$ admits an
infinite reduction. If this reduction contains finitely many $\redrroot{R}$-steps, then this
reduction ends in an infinite $\to_S$-reduction, contradicting the assumption that all terms
in this reduction are $S$-terminating. So this reduction contains infinitely many 
$\redrroot{R}$-steps, hence can be written as an infinite $(\to_S^* \cdot \redrroot{R}) \cap
\to_{\SN(\to_S)}$ reduction.  \qed
\end{proof}

\begin{lemma}
\label{lemdp2}
Let $R,S$ be TRSs. Then $\SNdpm{R}{S}$ holds if and only if for every term $t$ and
every $m \in \nat$ there exists $n \in \nat$ such that
\begin{quote}
for every $n$-step $(\redrroot{R} \cup \to_S)$-reduction $t=t_0 \to t_1 \to \cdots \to t_n$ there
exists $i \in [0,n]$ such that $t_i$ admits an $m$-step $\to_S$-reduction.
\end{quote}
\end{lemma}
\begin{proof}
Due to Lemma \ref{lemdp1} $\SNdpm{R}{S}$ is equivalent to finiteness of all 
$(\redrroot{R} \cup \to_S)$-reductions only consisting of $\to_S$-terminating terms.
Since $(\redrroot{R} \cup \to_S)$ is finitely branching, this is equivalent to 
\begin{quote}
for every term $t$ there exists $n \in \nat$ such that
no $n$-step $(\redrroot{R} \cup \to_S)$-reduction $t=t_0 \to t_1 \to \cdots \to t_n$ 
exists for which $t_i$ is $\to_S$-terminating for every $i \in [0,n]$.
\end{quote}
Since $\to_S$ is finitely branching, $\to_S$-termination of $t_i$ for every $i \in [0,n]$
is equivalent to the existence of $m \in \nat$ such that no $t_i$ admits an $m$-step
$\to_S$-reduction. After removing double negations, this proves equivalence with the claim
in the lemma.
\qed
\end{proof}

\begin{theorem}
The property $\SNdpm{R}{S}$ for given TRSs $R,S$ is $\cpi{0}{2}$-complete.
\end{theorem}

\begin{proof}
$\SN(R)$ is $\cpi{0}{2}$-complete and  $\SN(R)$ is 
equivalent to $\SNdpm{\DP(R)}{R}$, so $\SNdpm{R}{S}$ is $\cpi{0}{2}$-hard. That $\SNdpm{R}{S}$ is in $\cpi{0}{2}$ follows from Lemma
\ref{lemdp2}; note that the body of the claim in Lemma \ref{lemdp2} is recursive.
\qed
\end{proof}

\section{Conclusion and Future work}

In this paper we have analyzed the proof theoretic complexity, in term
of the arithmetic and analytical hierarchy, of termination properties
in term rewriting. The position of $\WN$ and $\SN$ were to be
expected, but the position of dependency pair problems is remarkably
high.
We have shown that (ground) confluence is $\cpi{0}{2}$-complete both uniform and for single terms.
The situation becomes more interesting when we look at weak confluence and
weak confluence on ground terms.
While the former is $\csig{0}{1}$,
the latter turns out to be $\cpi{0}{2}$-complete.
In future work, we will also further study the place in the analytic
hierarchy of properties of infinitary rewriting like $\WNi$.

\bibliography{main}

\end{document}